\documentclass[11pt,twoside]{article}
\usepackage{acta-info}
\usepackage{euler}

\newcommand{\vol}{\textbf{3,} 1 (2011) 35--47}   
\newcommand{\amss}{\amssubj{%
68Q15
}}     
\newcommand{\CRs}{\CRsubj{%
F.1.3
}}   
\newcommand{\keyw}{\keywords{%
longest common subsequence, sequence annotation, {\bf NP}-complete }}

\begin{document}

\title{Arc-preserving subsequences of arc-annotated sequences}
\maketitle

\oneauthor{%
\href{http://ai.math.usu.ru/site/}{Vladimir Yu. POPOV}
}{%
\href{http://www.usu.ru}{Department of Mathematics and Mechanics} \\
Ural State University \\ 620083 Ekaterinburg, RUSSIA
}{%
 \href{mailto:Vladimir.Popov@usu.ru}{Vladimir.Popov@usu.ru}
}

\short{%
V. Popov
}{%
Arc-preserving subsequences}

\begin{abstract}
Arc-annotated sequences are useful in representing the structural
information of RNA and protein sequences. The longest
arc-pre\-ser\-ving common subsequence problem has been introduced
as a framework for studying the similarity of arc-annotated
sequences. In this paper, we consider arc-annotated sequences with
various arc structures. We consider the longest arc preserving common subsequence 
problem. In particular, we show 
that the decision version of the $1$-{\sc
fragment  LAPCS(cross\-ing,cha\-in)} and the decision version of the $0$-{\sc diagonal
LAPCS(cross\-ing,cha\-in)} are {\bf NP}-complete for some fixed alphabet $\Sigma$ such that
$\vert \Sigma \vert = 2$.
Also we show that if $\vert \Sigma \vert = 1$, then the decision version of the $1$-{\sc
fragment  LAPCS(un\-li\-mi\-ted, plain)} and the decision version of the $0$-{\sc diagonal
LAPCS(un\-li\-mi\-ted, plain)} are {\bf NP}-complete.
\end{abstract}


\section{Introduction}

Algorithms on sequences of symbols have been studied for a long time
and now form a fundamental part of computer science. One of the very
important problems in analysis of sequences is the longest common
subsequence ({\sc LCS}) problem. The computational problem of
finding the longest common subsequence of a set of $k$ strings has
been studied extensively over the last thirty years (see
\cite{BDFW,Hir83,IF92} and references). This problem has many
applications. When $k=2$, the longest common subsequence is a
measure of the similarity of two strings and is thus useful in
molecular biology, pattern recognition, and text compression
\cite{LF78,Mai78,San72}. The version of {\sc LCS} in which the
number of strings is unrestricted is also useful in text compression
\cite{Mai78}, and is a special case of the multiple sequence
alignment and consensus subsequence discovery problem in molecular
biology \cite{DM93a,DM93b,Pev92}.

The $k$-unrestricted {\sc LCS} problem is {\bf NP}-complete
\cite{Mai78}. If the number of sequences is fixed at $k$ with
maximum length $n$, their longest common subsequence can be found in
$O(n^{k-1})$ time, through an extension of the pairwise algorithm
\cite{IF92}. Suppose $\vert S_1 \vert =n$  and $\vert S_2 \vert =m$,
the longest common subsequence of $S_1$   and $S_2$ can be found in
time $O(nm)$ \cite{CLRS,Hir,WF}.

Sequence-level investigation has become essential in modern
molecular bi\-o\-lo\-gy. But to consider genetic molecules only as
long sequences consisting of the 4 basic constituents is too simple
to determine the function and physical structure of the molecules.
Additional information about the sequences should be added to the
sequences. Early works with these additional information are primary
structure based, the sequence comparison is basically done on the
primary structure while trying to incorporate secondary structure
data \cite{BMR,CM}. This approach has the weakness that it does not
treat a base pair as a whole entity. Recently, an improved model was
proposed \cite{Ev1,Ev2}.

Arc-annotated sequences are useful in describing the secondary and
tertiary structures of RNA and protein sequences. See
\cite{Ev1,N13,G,JLMZ,N11} for further discussion and references.
Structure comparison for RNA and for protein sequences has become a
central computational problem bearing many challenging computer science
questions. In this context, the longest arc preserving common subsequence 
problem ({\sc LAPCS}) recently has received considerable attention
\cite{Ev1,Ev2,JLMZ,N11,LCJW}. It is a sound and meaningful mathematical formalization
of comparing the secondary structures of molecular sequences. Studies 
for this problem have been undertaken in \cite{BDFW,G,N2,N3,N1,N4,N5,N6,N7,N12,N8,N9,N10}.

\section{Preliminaries and problem definitions}

Given two sequences $S$ and $T$ over some fixed alphabet $\Sigma$,
the sequence $T$ is a subsequence of $S$ if $T$ can be obtained from
$S$ by deleting some letters from $S$. Notice that the order of the
remaining letters of $S$ bases must be preserved. The length of a
sequence $S$ is the number of letters in  it  and is denoted as
$\vert S \vert$. For simplicity, we use $S[i]$ to denote the $i$th
letter in sequence $S$, and $S[i,j]$ to denote the substring of $S$
consisting of the $i$th letter through the $j$th letter.

Given two sequences $S_1$ and $S_2$ (over some fixed alphabet
$\Sigma$), the classic longest common subsequence  problem asks for
a longest sequence $T$ that is a subsequence of both $S_1$ and
$S_2$.

An arc-annotated sequence of length $n$ on a finite alphabet $\Sigma$ is a couple
$A=(S,P)$ where $S$ is a sequence of length $n$ on $\Sigma$ and $P$ is a set of pairs
$(i_1,i_2)$, with $1 \leq i_1 < i_2 \leq n$.
In
this paper we will then call
an element of $S$ a base. A pair $(i_1,i_2) \in P$ represents an arc linking bases
$S[i_1]$ and $S[i_2]$ of $S$. The bases $S[i_1]$ and $S[i_2]$ are said to belong to the arc 
$(i_1,i_2)$ and are the only bases that belong to this arc.

Given two annotated sequences $S_1$ and $S_2$ with arc sets $P_1$
and $P_2$ respectively, a common subsequence $T$ of $S_1$ and $S_2$
induces a bijective mapping from a subset of
$\{ 1, \dots , \vert
S_1 \vert \}$ to subset of
$\{ 1, \dots , \vert S_2 \vert \}$.
The common subsequence $T$ is arc-preserving if the arcs induced by
the mapping are preserved, i.e., for any $( i_1, j_1)$ and $(i_2,
j_2)$ in the mapping,
$$(i_1,i_2) \in P_1 \Leftrightarrow (j_1,j_2) \in
P_2.$$

The {\sc LAPCS} problem 
is to find a longest common subsequence of $S_1$ and $S_2$ that is
arc-preserving (with respect to the given arc sets $P_1$ and $P_2$)
\cite{Ev1}.

{\sc { LAPCS}}:

{\sc Instance:} An alphabet $\Sigma$, annotated sequences $S_1$ and
$S_2$, $S_1, S_2 \in {\Sigma}^{\ast}$, with arc sets $P_1$ and $P_2$
respectively.

{\sc Question:} Find a longest common subsequence of $S_1$ and $S_2$
that is arc-preserving.

The arc structure can be restricted. We consider the following four
natural restrictions on an arc set $P$ which are first discussed in
\cite{Ev1}:

1. no sharing of endpoints:\\
\phantom{4.......} $\forall (i_1,i_2), (i_3,i_4)   \in P, i_1 \not=
i_4, i_2 \not= i_3,$ and $i_1=i_3 \Leftrightarrow i_2=i_4$.

2. no crossing:\\
\phantom{4.......} $\forall (i_1,i_2), (i_3,i_4)   \in P, i_1 \in
[i_3,i_4] \Leftrightarrow i_2 \in [i_3,i_4]$.

3. no nesting:\\
\phantom{4.......} $\forall (i_1,i_2), (i_3,i_4)   \in P, i_1 \leq
i_3 \Leftrightarrow i_2 \leq i_3$.

4. no arcs:\\
\phantom{4.......} $P= \emptyset$.

These restrictions are used progressively and inclusively to produce
five distinct levels of permitted arc structures for {\sc LAPCS}:

-- {\sc unlimited} --- no restrictions;

-- {\sc crossing} --- restriction 1;

-- {\sc nested} --- restrictions 1 and 2;

-- {\sc chain} --- restrictions 1, 2 and 3;

-- {\sc plain} --- restriction 4.

The problem {\sc LAPCS} is varied by these different levels of
restrictions as {\sc LAPCS}$(x,y)$ which is problem {\sc LAPCS} with
$S_1$ having restriction level $x$ and $S_2$ having restriction
level $y$. Without loss of generality, we always assume that $x$ is
the same level or higher than $y$.

We give the definitions of two special cases of the {\sc LAPCS}
problem, which were first studied in \cite{LCJW}. The special cases
are motivated from biological applications \cite{Gus,LMW}.

{\sc The $c$-fragment { LAPCS} problem $(c \geq 1)$}:

{\sc Instance:} An alphabet $\Sigma$, annotated sequences $S_1$ and
$S_2$, $S_1, S_2 \in {\Sigma}^{\ast}$,
 with arc sets
$P_1$ and $P_2$ respectively, where $S_1$ and $S_2$ are divided into
fragments of lengths exactly $c$ (the last fragment can have a
length less than $c$).

{\sc Question:} Find a longest common subsequence of $S_1$ and $S_2$
that is arc-preserving. The allowed matches are those between
fragments at the same location.

The $c$-{\sc diagonal LAPCS} problem, $(c \geq 0)$, is an extension
of the $c$-{\sc fragment  LAPCS} problem, where base $S_2[i]$ is
allowed only to match bases in the range $S_1[i-c, i+c]$.

The $c$-{\sc diagonal LAPCS} and $c$-{\sc fragment  LAPCS} problems
are relevant in the comparison of conserved RNA sequences where we
already have a rough idea about the correspondence between bases in
the two sequences.

\section{Previous results}

It is shown in \cite{LCJW} that the $1$-{\sc fragment
LAPCS(cross\-ing, cross\-ing)} and $0$-{\sc diagonal
LAPCS(cross\-ing, cross\-ing)} are solvable in time $O(n)$. 
An overview on known {\bf NP}-completeness results for
$c$-{\sc diagonal LAPCS} and $c$-{\sc fragment  LAPCS} is given in
Figure 1.

\noindent
 \begin{figure}[ht]
\begin{center}
\begin{tabular}{|c|c|c|c|c|c|}
\hline
                        &{unlimited}      &{crossing}       &{nested}&{chain}          &{plain}         \\
\hline
     {unlimited}    &{\bf NP}-h \cite{LCJW}&{\bf NP}-h \cite{LCJW}&{\bf NP}-h \cite{LCJW}&?&?\\
\hline
     {crossing}     &---                  &{\bf NP}-h \cite{LCJW}&{\bf NP}-h \cite{LCJW}&?&?\\
\hline
     {nested}       &---                  &---                  &{\bf NP}-h \cite{LCJW}&?&?\\
\hline
\end{tabular}
\end{center}
\caption{{\bf NP}-completeness results for $c$-{\sc diagonal LAPCS}
(with $c \geq 1$) and $c$-{\sc fragment  LAPCS} (with $c \geq 2$)}
 \end{figure}

\section{The $c$-{\sc fragment LAPCS(un\-li\-mi\-ted,plain)} 
and the $c$-{\sc  diagonal LAPCS(un\-li\-mi\-ted,plain)} problem}

Let us consider the decision version of the $c$-{\sc fragment LAPCS}
problem.

{\sc Instance:} An alphabet $\Sigma$, a positive integer $k$,
annotated sequences $S_1$ and $S_2$, $S_1, S_2 \in {\Sigma}^{\ast}$,
 with arc sets
$P_1$ and $P_2$ respectively, where $S_1$ and $S_2$ are divided into
fragments of lengths exactly $c$ (the last fragment can have a
length less than $c$).

{\sc Question:} Is there a common subsequence $T$ of $S_1$ and $S_2$
that is arc-preserving, $\vert T \vert \geq k$? (The allowed matches
are those between fragments at the same location).

Similarly, we can define the decision version of the $c$-{\sc
diagonal LAPCS} problem.

\begin{theorem}
 If $\vert \Sigma \vert = 1$, then $1$-{\sc
fragment  LAPCS(un\-li\-mi\-ted, plain)} and $0$-{\sc diagonal
LAPCS(un\-li\-mi\-ted, plain)} are {\bf NP}-complete.
\end{theorem}

\begin{proof}
 It is easy to see that $1$-{\sc fragment
LAPCS(un\-li\-mi\-ted, plain)} = $0$-{\sc diagonal
LAPCS(un\-li\-mi\-ted, plain)}.

Let $G=(V,E)$ be an undirected graph, and let $I \subseteq V$. We
say that the set $I$ is independent if whenever $i, j \in I$ then
there is no edge between $i$ and $j$. We make use of the following
problem:

{\sc Independent Set (IS)}: {\sc Instance:} A graph $G=(V,E)$, a
positive integer $k$.

{\sc Question:} Is there an independent set $I$, $I \subseteq V$,
with $\vert I \vert \geq k$?

{\sc IS} is {\bf NP}-complete (see \cite{Pap}).

Let us suppose that $\Sigma = \{ a \}$. We will show that {\sc IS}
can be polynomially reduced to problem $1$-{\sc fragment
LAPCS(un\-li\-mi\-ted, plain)}.

Let $\langle G=(V,E), V= \{ 1, 2, \dots , n \} , k \rangle$ be an
instance of {\sc IS}. Now we transform an instance of the {\sc IS}
problem to an instance of the $1$-{\sc fragment
LAPCS(un\-li\-mi\-ted, plain)} problem as follows.

$\bullet$ $S_1=S_2=a^n$.

$\bullet$ $P_1=E, P_2=\emptyset$.

$\bullet$ $\langle (S_1,P_1), (S_2,P_2) , k \rangle$.

First suppose that the graph $G$ has an independent set $I$ of size
$k$. By definition of independent set, $(i,j) \notin E$ for each $i,
j \in I$. For a given subset $I$, let
$$M= \{ (i, i) : i \in I \}.$$
Since $I$ is an independent set, if $(i,j) \in E=P_1$ then either
$(i,i) \notin M$ or $(j,j) \notin M$. This preserves arcs since
$P_2$ is empty. Clearly, $S_1[i]=S_2[i]$ for each $i \in I$, and the
allowed matches are those between fragments at the same location.
Therefore, there is a common subsequence $T$ of $S_1$ and $S_2$ that
is arc-preserving, $\vert T \vert =k$, and the allowed matches are
those between fragments at the same location.

Now suppose that there is a common subsequence $T$ of $S_1$ and
$S_2$ that is arc-preserving, $\vert T \vert =k$, and the allowed
matches are those between fragments at the same location. In this
case there is a valid mapping $M$, with $\vert M \vert =k$. Since
$c=1$, it is easy to see that if $(i,j) \in M$ then $i=j$. Let
$$I=
\{ i : (i,i) \in M \}.$$ Clearly,
$$\vert I \vert = \vert M \vert
=k.$$ Let $i_1$ and $i_2$ be any two distinct members of $I$. Then
let $(i_1,j_1), (i_2,j_2) \in M$. Since
$$i_1=j_1, i_2=j_2, i_1 \not=
i_2,$$ it is easy to see that $j_1 \not= j_2$. Since $P_2$ is empty,
$(j_1,j_2) \notin P_2$, so $(i_1,i_2) \notin P_1$. Since $P_1=E$,
the set $I$ of vertices is a size $k$ independent set of $G$.
\end{proof}

\section{The $c$-{\sc fragment  LAPCS(cross\-ing,cha\-in)} 
and the $c$-{\sc di\-a\-go\-nal  LAPCS(cross\-ing,cha\-in)} problem}

\begin{theorem} If $\vert \Sigma \vert = 2$, then $1$-{\sc
fragment  LAPCS(cross\-ing,cha\-in)} and $0$-{\sc diagonal
LAPCS(cross\-ing,cha\-in)} are {\bf NP}-complete.
\end{theorem}

\begin{proof}
 It is easy to see that $1$-{\sc fragment
LAPCS(cross\-ing, chain)} = $0$-{\sc diagonal LAPCS(cross\-ing,
chain)}.

Let us suppose that $\Sigma = \{ a, b \}$. We will show that {\sc
IS} can be polynomially reduced to problem $1$-{\sc fragment
LAPCS(cross\-ing, chain)}.

Let $\langle G=(V,E), V= \{ 1, 2, \dots , n \} , k \rangle$ be an
instance of {\sc IS}. Note that {\sc IS} remains {\bf NP}-complete
when restricted to connected graphs with no loops and multiple
edges. Let $G=(V,E)$ be such a graph. Now we transform an instance
of the {\sc IS} problem to an instance of the $1$-{\sc fragment
LAPCS(cross\-ing, chain)} problem as follows.

\newpage 
\noindent
There are two cases to consider.

\medskip
\noindent {\bf Case I}. $k>n$

$\bullet$ $S_1=S_2=a$

$\bullet$ $P_1=P_2=\emptyset$

$\bullet$ $\langle (S_1,P_1), (S_2,P_2) , k \rangle$

Clearly, if $I$ is an independent set, then $I
\subseteq V$ and $\vert I \vert \leq \vert V \vert =n$. Therefore,
there is no an independent set $I$, with $\vert I \vert \geq k$.

Since $k>n$ and $n \in \{ 1, 2, \dots  \}$, it is easy to see that
$k>1$. Since $S_1=S_2=a$ and $P_1=P_2=\emptyset$, $T=a$ is the
longest arc-preserving common subsequence. Therefore, there is no an
arc-preserving common subsequence $T$ such that $\vert T \vert \geq
k$.

\medskip\noindent {\bf Case II}. $k \leq n$

$\bullet$ $S_1=S_2=(ba^{n}b)^n$

$\bullet$ Let $\alpha < \beta$. Then

$$( \alpha , \beta ) \in P_1 \Leftrightarrow
[ \exists i \in \{ 1, 2, \dots , n \} \exists j \in \{ 1, 2, \dots ,
n \}$$
$$((i,j) \in E \wedge  \alpha =(i-1)(n+2)+j+1 \wedge$$
$$\wedge \beta = (j-1)(n+2)+i+1)] \vee$$
$$\vee [ \exists i  \in \{ 1, 2, \dots , n \}
( \alpha = (i-1)(n+2)+1 \wedge \beta = i(n+2))],$$

$$( \alpha , \beta ) \in P_2 \Leftrightarrow
\exists i  \in \{ 1, 2, \dots , n \}$$
$$( \alpha = (i-1)(n+2)+1
\wedge \beta = i(n+2)).$$

$\bullet$ $\langle (S_1,P_1), (S_2,P_2) , k(n+2) \rangle$

First suppose that  $G$ has an independent
set $I$ of size $k$. By definition of independent set, $(i,j) \notin
E$ for each $i, j \in I$. For a given subset $I$, let
$$M= \{ (j, j) : j=(n+2)(i-1) +l, i \in I , $$
$$l \in \{ 1, 2, \dots , n+2 \} \} .$$

Let $(j,j) \in M$, and there exist $i$ such that $j=(n+2)(i-1) +1$.
By definition of $M$,
$$((n+2)(i-1) +1,(n+2)(i-1) +1) \in M \Leftrightarrow$$
$$\Leftrightarrow
((n+2)i,(n+2)i) \in M.$$ By definition of $P_l$, $((n+2)(i-1)
+1,(n+2)i) \in P_l$ where $l=1,2$. Let $(j,j) \in M$, and there
exist $i$ such that $j=(n+2)i$. By definition of $M$,
$$((n+2)i,(n+2)i) \in M \Leftrightarrow$$
$$\Leftrightarrow
((n+2)(i-1) +1,(n+2)(i-1) +1) \in M.$$ By definition of $P_l$,
$$((n+2)(i-1) +1,(n+2)i) \in P_l$$
where $l=1,2$. Let $(j,j) \in M$, and
$$j=(n+2)(i-1) +l$$
where $1<l<n+2$. By definition of $M$, $i \in I$. Since $I$ is an
independent set, if $(i,l-1) \in E$ then $l-1 \notin I$. Since
$$1<l<n+2,$$ by definition of $P_1$, either
$$((n+2)(i-1) +l,(n+2)(l-2) +i+1) \in P_1$$
or
$$((n+2)(i-1) +l,t)
\notin P_1$$ for each $t$. Since
$$1<l<n+2,$$
by definition of $P_2$,
$$((n+2)(i-1) +l,t) \notin P_2$$
for each $t$. If
$$((n+2)(i-1)
+l,(n+2)(l-2) +i+1) \in P_1,$$ then in view of $l-1 \notin I$,
$$((n+2)(l-2) +i+1,(n+2)(l-2) +i+1) \notin M.$$
This preserves arcs. Since $\vert I \vert =k$, it is easy to see
that
$$\vert M \vert
=k(n+2).$$ Clearly, $S_1[i]=S_2[i]$ for each $i \in I$, and the
allowed matches are those between fragments at the same location.
Therefore, there is a common subsequence $T$ of $S_1$ and $S_2$ that
is arc-preserving, $\vert T \vert =k(n+2)$, and the allowed matches
are those between fragments at the same location.

Now suppose that there is a common subsequence $T$ of $S_1$ and
$S_2$ that is arc-preserving, $\vert T \vert =k$, and the allowed
matches are those between fragments at the same location. In this
case there is a valid mapping $M$, with $\vert M \vert =k$. Since
$c=1$, it is easy to see that if $(i,j) \in M$ then $i=j$. Let $I=
\{ i : (i,i) \in M \}$. Clearly, $\vert I \vert = \vert M \vert =k$.
Let $i_1$ and $i_2$ be any two distinct members of $I$. Then let
$(i_1,j_1), (i_2,j_2) \in M$. Since $i_1=j_1, i_2=j_2, i_1 \not=
i_2$, it is easy to see that $j_1 \not= j_2$. Since $P_2$ is empty,
$(j_1,j_2) \notin P_2$, so $(i_1,i_2) \notin P_1$. Since $P_1=E$,
the set $I$ of vertices is a size $k$ independent set of $G$.
\end{proof}

\section{Conclusions}

In this paper, we considered  two special cases of the {\sc LAPCS}
problem, which were first studied in \cite{LCJW}. We have shown 
that the decision version of the $1$-{\sc
fragment  LAPCS(cross\-ing,cha\-in)} and the decision version of the $0$-{\sc diagonal
LAPCS(cross\-ing,cha\-in)} are {\bf NP}-complete for some fixed alphabet $\Sigma$ such that
$\vert \Sigma \vert = 2$.
Also we have shown that if $\vert \Sigma \vert = 1$, then the decision version of the $1$-{\sc
fragment  LAPCS(un\-li\-mi\-ted, plain)} and the decision version of the $0$-{\sc diagonal
LAPCS(un\-li\-mi\-ted, plain)} are {\bf NP}-complete.
This results answers some open questions in \cite{G} (see Table 4.2. in \cite{G}).

\section*{Acknowledgements}

The work was partially supported by Grant of President of the Russian
Federation MD-1687.2008.9 and Analytical Departmental Program
  ``Developing the scientific potential of high school'' 2.1.1/1775.



\bigskip
\rightline{\emph{Received:  November 17, 2010{\tiny \raisebox{2pt}{$\bullet$\!}} Revised: March 11, 2011}} 

\end{document}